\documentclass[twocolumn,aps,pra,reprint,superscriptaddress,amsmath,amssymb,bbm,floatfix]{revtex4-2}

\usepackage[utf8]{inputenc}
\usepackage[T1]{fontenc}
\usepackage{color}
\usepackage{amstext}
\usepackage{graphicx}
\usepackage[colorlinks, allcolors={blue}]{hyperref}
\usepackage{physics}
\usepackage{blindtext}
\usepackage{pifont}
\usepackage{xcolor,graphicx}
\usepackage{newlfont}
\usepackage{amssymb,amsmath,mathrsfs,amsthm}
\usepackage{verbatim}
\usepackage{bbm}
\usepackage{multirow}
\usepackage[varg]{txfonts}
\usepackage{quantikz}

\newcommand{\orcid}[1]{\href{https://orcid.org/#1}{\includegraphics[width=7pt]{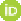}}}

\newcommand{\be}{\begin{equation}}
\newcommand{\ee}{\end{equation}}
\newcommand{\beq}{\begin{eqnarray}}
\newcommand{\eeq}{\end{eqnarray}}

\newcommand{\mc}[1]{\mathcal{#1}}

\renewcommand\bra[1]{{\langle{#1}|}}
\renewcommand\ket[1]{{|{#1}\rangle}}
\renewcommand{\Tr}{\text{Tr}}
\newcommand{\I}{\mathbbm{1}}
\newcommand{\im}{\mathbbm{i}}

\newtheorem{theorem}{Theorem}

\newtheorem{corollary}{Corollary}

\theoremstyle{remark}

\newtheorem{definition}{Definition}

\allowdisplaybreaks

\usepackage[markup=default,    
            authormarkup=none, 
            colorlinks=false,commandnameprefix=ifneeded]{changes}
\definecolor{addedcolor}{rgb}{0.0,0.0,1.0}   
\definecolor{deletedcolor}{rgb}{0.8,0.0,0.0} 
\setaddedmarkup{\textcolor{addedcolor}{#1}}
\setdeletedmarkup{\textcolor{deletedcolor}{\sout{#1}}}

\begin{document}

\title{Optimality of universal conclusive entanglement concentration protocols}

\author{Alexandre C. Orthey, Jr.\orcid{0000-0001-8111-3944}}
\email{alexandre.orthey@gmail.com}
\affiliation{Institute of Fundamental Technological Research, Polish Academy of Sciences, Pawi\'nskiego 5B, 02-106 Warsaw, Poland.}
\affiliation{Faculty of Mathematics, Informatics and Mechanics, University of Warsaw, ulica Banacha 2, 02-097 Warsaw, Poland}

\author{Aby Philip\orcid{0000-0002-4608-7522}}
\affiliation{Institute of Fundamental Technological Research, Polish Academy of Sciences, Pawi\'nskiego 5B, 02-106 Warsaw, Poland.}

\author{Tulja Varun Kondra\orcid{0000-0002-4228-6129}}
\affiliation{Institute for Theoretical Physics III, Heinrich Heine University D\"{u}sseldorf, Universit\"{a}tsstra{\ss}e 1, D-40225 D\"{u}sseldorf, Germany}

\author{Alexander Streltsov\orcid{0000-0002-7742-5731}}
\affiliation{Institute of Fundamental Technological Research, Polish Academy of Sciences, Pawi\'nskiego 5B, 02-106 Warsaw, Poland.}

\begin{abstract}
Entanglement concentration is essential for quantum technologies; yet, rigorous bounds on the success probability for universal protocols---those requiring no prior knowledge about the input state---have remained underexplored. We establish such fundamental limits for conclusive protocols distilling a perfect Bell state from pure two-qubit states by deriving the optimal success probability starting with: two copies of a state with known Schmidt basis, and four copies of a state with unknown Schmidt basis, using concatenated two-qubit operations. We prove that a known protocol achieves these bounds, confirming its optimality. Crucially, universality imposes an inherent efficiency trade-off, yielding an average success probability of just $2/105$ over Haar measure.
\end{abstract}

\maketitle

\section{Introduction}

Entanglement is a fundamental resource in quantum information processing \cite{horodecki2009quantum}, allowing for applications such as quantum teleportation \cite{bennett1993teleporting}, cryptography \cite{ekert1991}, and computation \cite{Shor2006Jul}. A critical challenge in practical settings is the \textit{concentration} of entanglement: distilling near-perfect Bell states $\ket{\phi^{+}} = (\ket{00} + \ket{11})/\sqrt{2}$ from multiple copies of partially entangled states shared between distant parties. While protocols exist for this task \cite{Bennett1996Jan,Deutsch1996Sep,Bennett1996Nov,Dur2007Jul,Torres2016Nov,Rozpedek2018Jun,Preti2022Aug,Miguel-Ramiro2025Apr} with several experimental implementations \cite{Pan2001Apr,Pan2003May,Hu2021Jan,Ecker2021Jul,reichle2006experimental,Kalb2017Jun}, their \textit{universality}, i.e., effectiveness across a large set of unknown input states, and \textit{optimality}, that is, if they produce $\ket{\phi^+}$ with the highest possible success probability, remain not well explored in quantum information theory. An exception is the result from Ref.~\cite{Zang2025May}, a no-go theorem stating that there is no LOCC (local operations and classical communication) protocol capable of generating a state with higher fidelity than its input states for all two-qubit states, including mixed ones. 

This work establishes rigorous optimality bounds for universal conclusive entanglement concentration protocols (ECPs), strictly for finite copies of pure input states and two-qubit operations. Unlike faithful protocols \cite{Vidal2000Jun} that deterministically produce a state with higher fidelity with the target state, \textit{conclusive} ECPs \cite{Vidal1999entanglement} use probabilistic LOCC  maps to either produce a perfect target state or fail, with no intermediate outcomes. By considering a smaller set of states and not allowing for intermediate outcomes, we avoid contradiction with the no-go theorem in \cite{Zang2025May}.  

In this work, we are interested in producing $\ket{\phi^+}$ outright with some probability, rather than in just increasing the fidelity of the input states with $\ket{\phi^+}$. The core challenge addressed here is: for $n=2$ and $n=4$, what is the optimal (or maximal) success probability for distilling a Bell pair from $n$ copies of an arbitrary pure two-qubit state, when no prior knowledge of the state's entanglement structure is available?  

We present a comprehensive optimality analysis for universal conclusive ECPs, considering two scenarios: (i) states with a known Schmidt basis, and (ii) completely arbitrary states with an unknown basis. For the first case, we prove the optimal probability of distilling a Bell state from two copies is a simple function of the entanglement of the input state. For the second, more challenging case, we derive a tight upper bound on the probability of transforming two copies into an entangled state within a fixed Schmidt basis. Furthermore, we prove a tight bound for the direct distillation of a Bell state from four copies of an arbitrary pure state, using protocols restricted to concatenated two-qubit operations.

The structural framework of concatenated two-qubit operations is motivated by both fundamental and practical considerations. Physically, since the input states possess a completely unknown Schmidt basis while the target Bell state has a fixed basis, any universal protocol must eventually bridge this gap; our concatenated approach explicitly isolates and optimally solves this necessary intermediate step. Furthermore, from an experimental standpoint, native entangling operations in near-term quantum architectures are predominantly limited to two qubits. This can be seen in several experimental realizations of entanglement distillation~\cite{reichle2006experimental,chen2017experimental,Kalb2017Jun,Hu2021Jan,Zhou2025Jul}. Restricting our analysis to two-qubit interactions avoids the prohibitive complexity and noise associated with global four-qubit joint measurements, ensuring our derived limits and protocols are directly relevant for practical implementations.

Our analysis employs a unified framework, established by Definitions \ref{def:UniversalCECP} and \ref{def:OptimalUCECP} below, that formalizes universal conclusive ECPs as LOCC maps requiring unit fidelity to $\phi^{+}$, while maximizing success probability across all inputs. For unknown Schmidt basis states, we derive tight upper bounds on success probabilities by imposing algebraic constraints on Kraus operators that enforce universality through nullification of separable subspaces. We then construct explicit LOCC protocols with operators that saturate these bounds. Our results show that K\'alm\'an et al.'s \cite{Kalman2025Dec} protocol is an optimal universal ECP in both scenarios. While Vidal's formula \cite{Vidal1999entanglement} gives higher conversion probabilities, as one is allowed to tailor the protocol to the specific input, we prove that universal protocols necessarily achieve strictly lower probabilities due to the universality constraint. Finally, we compute average performance using Dirichlet distributions for Haar-random states via moment formulae \cite{Zyczkowski2001Aug,bengtsson2017geometry,Schiavo2019Jan}.

\section{Optimality of a Universal Conclusive ECP}

Let $\rho_{AB}$ be a bipartite state shared between Alice and Bob. Let $\mathcal{E}_1$ and $\mathcal{E}_2$ be two completely positive trace non-increasing LOCC maps, such that $\mathcal{E}=\mathcal{E}_1+\mathcal{E}_2$ is a completely positive trace-preserving map. Having said that, we define:
\begin{definition}[Universal Conclusive ECP]\label{def:UniversalCECP}
        Let $\mathcal{S}$ be the subset of states $\rho$ that are shared between Alice and Bob. The $n$-to-$1$ map $\mathcal{E}_1$ is called a universal conclusive ECP for a given number of copies $n$ if, for every $\rho\in\mathcal{S}$, it holds that $\langle\phi^+\vert\mathcal{E}_1\left(\rho^{\otimes n}\right)\vert\phi^+\rangle=\operatorname{Tr}\left[\mathcal{E}_1\left(\rho^{\otimes n}\right)\right]$.
    \end{definition}
\begin{definition}[Optimal Universal Conclusive ECP]\label{def:OptimalUCECP}
        Let $\mathcal{S}$ be the subset of states $\rho$ that are shared between Alice and Bob. Let $\mathcal{E}_1$ be any protocol satisfying Def. \ref{def:UniversalCECP}. The map $\mathcal{E}_1^\text{optimal}$ is called the optimal universal conclusive ECP for a given number of copies $n$ if, for every $\rho\in\mathcal{S}$ and for any $\mathcal{E}_1$, it holds that $\operatorname{Tr}\left[\mathcal{E}_1\left(\rho^{\otimes n}\right)\right]\leqslant \operatorname{Tr}\left[\mathcal{E}_1^\text{optimal}\left(\rho^{\otimes n}\right)\right]$.
\end{definition}

First, we would like to note that there can be no universal conclusive ECP for the set of pure states shared between Alice and Bob when $n=1$. This follows from some simple arguments. Consider, toward a contradiction, that there exists a Universal Conclusive ECP for $n=1$ which means $\mathcal{E}_{1}(\vert\psi\rangle\langle\psi\vert)=p\vert\phi^+\rangle\langle\phi^+\vert$, where $p>0$. Let $\rho=(1-\epsilon)\,\vert\psi\rangle\langle\psi\vert + (\epsilon/4)\,\mathbb{I}$, then $\mathcal{E}_{1}(\rho)= p^{\prime}\vert\phi^+\rangle\langle\phi^+\vert$ where $p^{\prime}>0$. (Note, we are allowed to do this because $\rho$ can be expressed as a convex combination of pure states) This means that we will have transformed a full rank state to a pure entangled state by an LOCC protocol. However, this is not possible, since it is impossible to transform a full rank state into a pure resourceful state using free operations~\cite{PhysRevLett.125.060405,PhysRevA.101.062315}. Hence, we have arrived at a contradiction, and hence there is no Universal Conclusive ECP for the set of pure states shared between Alice and Bob and $n=1$. 

Similarly, using the aforementioned fact from~\cite{PhysRevLett.125.060405,PhysRevA.101.062315}, it is easy to see that there can be no universal conclusive ECP for the set of all states of a given dimension shared between Alice and Bob for any $n$. Hence, in this work, we shall only be looking at universal conclusive ECPs for sets of pure states.

\section{concentration protocol for states in a known Schmidt basis}

We will start by looking into universal conclusive ECPs over a particular set of pure states: Schmidt states. Let $\mathcal{S}_{2}$ be the set of bipartite pure states that have the same Schmidt basis. Without loss of generality, let the Schmidt basis under consideration be $\{\vert00\rangle_{AB},\vert11\rangle_{AB}\}$. Then, any state in $\mathcal{S}_{2}$ can be written as $\vert\psi_S\rangle_{AB} = \alpha\vert00\rangle_{AB} + \beta\vert11\rangle_{AB}$, where $\alpha,\beta\in\mathbb{C}$ satisfy $|\alpha|^2+|\beta|^2=1$. We will drop the subindex notation $AB$ whenever suitable. Also, for any pure state $\ket{\psi}$, we define $\psi\equiv\ket{\psi}\bra{\psi}$. Our first result is summarized in the following theorem:
\begin{theorem}\label{theo1}
    Following Def. \ref{def:OptimalUCECP}, the optimal universal conclusive ECP over the set of Schmidt states $\ket{\psi_S}=\alpha\ket{00}+\beta\ket{11}$ transforms two copies of such a state into $\phi^+$ with optimal probability $P_{\psi_S^{\otimes 2}\to\phi^+}\coloneqq 2|\alpha\beta|^2$.
\end{theorem}
\noindent\textit{Sketch of the proof.} The reader can find the details of the proof in Appendix \ref{appendix_A}. Since both input and output states are pure under the map $\mathcal{E}_{\phi^+}(\cdot)=\sum_i M_i(\cdot)M_i^\dagger$, the action of \textit{each} Kraus operator $M_i$ over the product $\ket{\psi_S}^{\otimes 2}$ must result into a state proportional to $\ket{\phi^+}_{AB}\ket{\text{garb}}_{A'B'}$, where $\ket{\text{garb}}_{A'B'}$ is some garbage state to be discarded. In addition to that, copies of product states such as $\ket{00}\ket{00}$ and $\ket{11}\ket{11}$ should not yield entangled states under LOCCs. These both conditions together impose an upper bound on the probability of obtaining $\phi^+$ as $2|\alpha\beta|^2$.  

Note that there can not be a universal conclusive ECP for the set of pure states shared between Alice and Bob when $n=2$. To see this, assume that such a protocol exists and apply it on the convex combination of any two states with different Schmidt bases, such as $\ketbra{\psi_S}^{\otimes 2}$ and $\ketbra{++}^{\otimes 2}$. Then we will have again transformed a full rank state to a pure entangled state by an LOCC protocol. However, this is not possible~\cite{PhysRevLett.125.060405,PhysRevA.101.062315}. Hence, we have arrived at a contradiction, and hence there is no Universal Conclusive ECP for the set of pure states shared between Alice and Bob and $n=2$.
    
\section{concentration protocol for states in an unknown Schmidt basis}

Now,  we consider concentration protocols given by the map $\mathcal{E}_{1}$ satisfying Def. \ref{def:OptimalUCECP} applied to four copies of a state $\ket{\psi}\in\mc{S}$, such that $\mc{S}$ is the set of all pure two-qubit states. For this case, we are not going to assume a known Schmidt basis, and therefore, we can write $\ket{\psi}$ in the computational basis as
    \begin{equation}\label{psi_c1234}
        \ket{\psi}=c_1\ket{00}+c_2\ket{01}+c_3\ket{10}+c_4\ket{11},
    \end{equation}
where $c_i\in\mathbb{C}$ for every $i$ and $\sum_{i=1}^4|c_i|^2=1$. By Def. \ref{def:OptimalUCECP}, we have that $\mathcal{E}_1(\psi^{\otimes 4})=P_{\psi^{\otimes 4}\to\phi^+}\,\phi^+$, where $P_{\psi^{\otimes 4}\to\phi^+}$ is the probability of success given by the renormalization required by some measurement included in the protocol, i.e., $P_{\psi^{\otimes 4}\to\phi^+}=\Tr\left[\mathcal{E}_{1}(\psi^{\otimes 4})\right]$. 

From Theorem \ref{theo1}, we know the optimal probability of obtaining state $\ket{\phi^+}$ from two copies of $\ket{\psi_S}$. However, that requires that we know the Schmidt basis before the beginning of the protocol. To use arbitrary states \eqref{psi_c1234} as input states, we can find a universal protocol that converts two copies of $\ket{\psi}$ into a state with a known Schmidt basis $\vert\psi_{S}\rangle$, and then we can apply the protocol that yields Theorem \ref{theo1} onto two copies of $\psi_S$ to obtain state $\phi^+$. As we are going to show, the same protocol that produces Schmidt states is also the protocol that satisfies Theorem \ref{theo1} (see Fig. \ref{fig:diagram_4_copies}). Note that, in principle, the coefficients $\alpha$ and $\beta$ of the output state are not specified; what matters is that the protocol must universally produce a state in a known Schmidt basis.

\begin{figure}[h]
    \centering
    \includegraphics[width=0.9\linewidth]{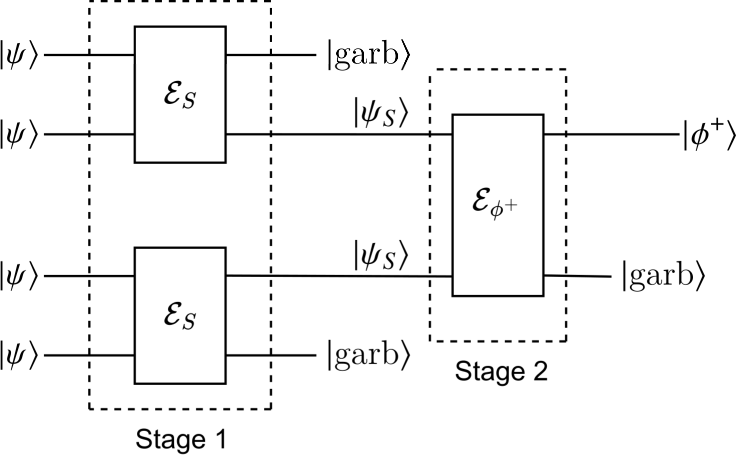}
    \caption{Optimal universal conclusive ECP applied to four copies of an unknown two-qubit state $\ket{\psi}$. In the first stage of the protocol, four copies of $\ket{\psi}$ are transformed into two copies of $\ket{\psi_S}$, a state with a known Schmidt basis. In the second stage, two copies of $\ket{\psi_S}$ are transformed into one copy of $\ket{\phi^+}$.}
    \label{fig:diagram_4_copies}
\end{figure}

\section{Producing entangled states in a known Schmidt basis}

Let us look for the optimal probability of obtaining an entangled state in a known Schmidt basis from two copies of \eqref{psi_c1234} by using a universal conclusive protocol. We denote the completely positive trace non-increasing map as $\mathcal{E}_{S}$ and the corresponding Kraus operators as $M_i$, such that
\begin{equation}\label{Lambda_S}
    \mathcal{E}_{S}(\cdot)=\sum_i M_i (\cdot) M_i^\dagger,
\end{equation}
where $\sum_i M_iM_i^\dagger\leqslant \mathbb{I}$. As in the proof of Theorem \ref{theo1}, we know that product states must not produce any entanglement under LOCCs. Hence, we must have that
\begin{equation}\label{systemM}
    \begin{array}{cc}
        M_i\ket{0000}_{ABA'B'}=0,& M_i\ket{0101}_{ABA'B'}=0, \\
        M_i\ket{1010}_{ABA'B'}=0,& M_i\ket{1111}_{ABA'B'}=0.
    \end{array}
\end{equation}
for every $i$. From the above, we can see that $M_i$ must act such that
\begin{align}\label{M_ipsi1}
        M_i & \ket{\psi}^{\otimes 2}\nonumber\\
        =& M_i [   
         c_1c_2\left( \ket{0001}+\ket{0100}\right) + c_1c_3\left( \ket{0010}+\ket{1000}\right) \nonumber\\
        &+c_2c_4\left( \ket{0111}+\ket{1011}\right) + c_3c_4\left( \ket{1011}+\ket{1110}\right) \nonumber \\
        &+c_1c_4\left( \ket{0011}+\ket{1100}\right) + c_2c_3\left( \ket{0110}+\ket{1001}\right)]
\end{align}
where $\ket{jklm}\in\mathcal{H}_{ABA'B'}$ whenever $j,k,l,m\in\{0,1\}$. 
   
We can impose further constraints on the action of $M_{i}$. Consider now the case when the input state in~\eqref{psi_c1234} is such that $c_3=c_4=0$, which implies that the input state is a product state $\ket{\psi}=\ket{0}_A\otimes(c_1\ket{0}_B+c_2\ket{1}_B)$. We know that applying map $\mathcal{E}_{S}$ on copies of such a state must never produce an entangled state $\psi_S$. Using the same argument when $c_1=c_2=0$, $c_2=c_4=0$, or $c_1=c_3=0$, we infer that, from \eqref{M_ipsi1}, $M_{i}$ must satisfy the following equations:
\begin{equation}
    \begin{array}{cc}\label{systemM2}
        M_i\left(\ket{0001}+\ket{0100}\right) =0,& M_i\left(\ket{0010}+\ket{1000}\right) =0,\\
        M_i\left(\ket{0111}+\ket{1011}\right) =0,& M_i\left(\ket{1011}+\ket{1110}\right) =0,
    \end{array}
\end{equation}
for all $i$. Now, we can use constraints \eqref{systemM} and \eqref{systemM2} to help us find the explicit form of the Kraus operators that can do the transformation $\psi^{\otimes 2}\to\psi_S$. In fact, each operator $M_i$ must be a product of two identical Kraus operators acting on spaces $\mc{H}_{AA'}$ and $\mc{H}_{BB'}$, that is
\begin{equation}\label{Mi=KiKi}
    (M_i)_{AA'BB'}=(K_i)_{AA'}\otimes (K_i)_{BB'}, \qquad \forall i.
\end{equation}
This symmetry is necessary because of the universality requirement of our protocol; since we do not know how entanglement is distributed between Alice and Bob, the protocol should treat each party equally. From \eqref{systemM}, \eqref{systemM2}, and \eqref{Mi=KiKi}, we can state our next result:
\begin{theorem}\label{theo2}
    To produce arbitrary states in a known Schmidt basis from two copies of arbitrary states $\ket{\psi}_{AB}$ given by \eqref{psi_c1234}, a map $\mathcal{E}_S(\cdot)=\sum_i M_i(\cdot) M_i^\dagger$ must be composed of Kraus operators $M_i=K_i\otimes K_i$ acting on $\mathcal{H}_{AA'}\otimes\mathcal{H}_{BB'}$ such that
    \begin{equation}\label{K_i_main}
        K_i=a(\ket{00}\bra{01} + \ket{00}\bra{10})+b(\ket{10}\bra{01} - \ket{10}\bra{10})
    \end{equation}
    for every $i$, where $a$ and $b$ are any non-zero complex numbers satisfying $2(|a|^4+|b|^4)\leqslant 1$. The success probability $P_{\psi^{\otimes 2}\to\psi_S}$ of such a protocol is bounded by $\overline{P}_{\psi^{\otimes 2}\to\psi_S}=2(|c_1c_4|+|c_2c_3|)^2$, i.e., $P_{\psi^{\otimes 2}\to\psi_S}<\overline{P}_{\psi^{\otimes 2}\to\psi_S}$.
\end{theorem}
\noindent\textit{Sketch of the proof.} The reader can find the details of the proof in Appendix \ref{appendix_B}. From \eqref{systemM} and \eqref{systemM2}, we find that $M_i$ must act on $\ket{\psi}^{\otimes 2}$ such that $M_i\ket{\psi}^{\otimes 2}= \sqrt{2} M_i (c_1c_4\ket{\phi_1}+c_2c_3\ket{\phi_2})$, where we introduced the notation $\ket{\phi_1}\coloneqq(\ket{0011}+\ket{1100})/\sqrt{2}$ and $\ket{\phi_2}\coloneqq(\ket{0110}+\ket{1001})/\sqrt{2}$. It follows that $M_i$ applied to $\ket{\phi_j}$ must result in $\ket{\psi_S^{(j)}}=\alpha_j\ket{00}+\beta_j\ket{11}$ with some probability $p_{i,j}$ for $j=1,2$. The linear combination of these Schmidt states also configures a Schmidt state. Then, we can write the Kraus operators using a Pauli basis, i.e., $K_i=\sum_{kl} r_{kl}\sigma_k\otimes\sigma_l$, and impose conditions \eqref{systemM} and \eqref{systemM2}. Normalization conditions then give us the constraint $2(|a|^4+|b|^4)\leqslant 1$, where $a=r_{14}/4$ and $b=r_{34}/4$. Following this constraint, we can directly calculate the success probability $P_{\psi^{\otimes 2}\to\psi_S}$ and derive the upper bound $\overline{P}_{\psi^{\otimes 2}\to\psi_S}=2(|c_1c_4|+|c_2c_3|)^2$, where $\overline{P}_{\psi^{\otimes 2}\to\psi_S} - P_{\psi^{\otimes 2}\to\psi_S} > 4(1-f)|c_1c_2c_3c_4|$, $0<f<1$, and $f=2||a|^4-|b|^4|$.

Note that the success probability $P_{\psi^{\otimes 2}\to\psi_S}$ can reach the bound $\overline{P}_{\psi^{\otimes 2}\to\psi_S}$ if, and only if, $f=1$, which in turn requires $|a|=1/\sqrt[4]{2}$ and $b=0$ (or the opposite). However, Kraus operators \eqref{K_i_main} with such parameters can only generate product states (see Appendix \ref{appendix_B}), which would be useless in the second stage of the protocol. In Fig. \ref{fig:density}, we plot $f$ as a function of $|a|$ and $|b|$. Only in the dark regions, the success probability gets closer to the upper bound $\overline{P}_{\psi^{\otimes 2}\to\psi_S}$. 

\begin{figure}[h]
    \centering
    \includegraphics[width=\linewidth]{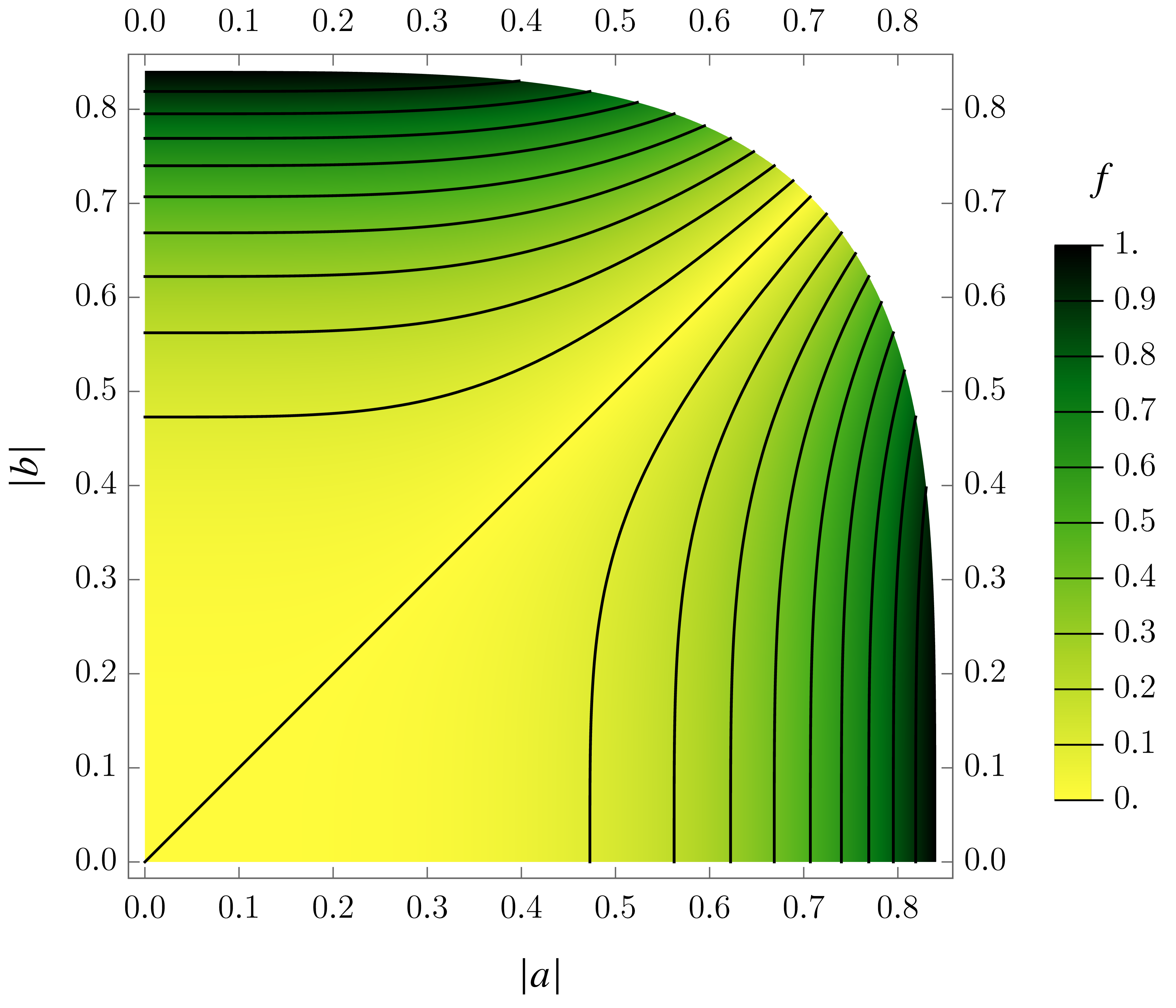}
    \caption{Domain of the absolute value of coefficients $a$ and $b$ for Kraus operators satisfying Theorem \ref{theo2}: $\{a,b\in\mathbb{C}\,\vert\, 2(|a|^4+|b|^4)\leqslant 1 \text{ and } (a\neq 0 \text{ or } b\neq 0)\}$. Function $f=2||a|^4-|b|^4|$ determines how much the success probability can get close to the bound $\overline{P}_{\psi^{\otimes 2}\to\psi_S}$; if $f\sim 1$, the bound is achieved.}
    \label{fig:density}
\end{figure}

\section{Optimal probability of producing $\phi^+$ from 4 copies of unknown two-qubit state}

Now, we shall combine Theorems \ref{theo1} and \ref{theo2} and derive the following result:

\begin{theorem}\label{theo3}
    The success probability of transforming 4 copies of $\ket{\psi}$ into $\phi^+$ using a universal conclusive ECP based on concatenated two-qubit operations is bounded by
    \begin{equation}\label{opt4copies}
    P_{\psi^{\otimes 4}\to\phi^+}\leqslant 2\left| c_2^4 c_3^4 \right| + 2\left| c_1^4 c_4^4 \right| -4\mathrm{Re}\left[ c_1^2 c_4^2 (c_2^*)^2 (c_3^*)^2\right].
    \end{equation}
\end{theorem}
\begin{proof}
First, note that the action of $M_i=K_i\otimes K_i$ from Theorem \ref{theo2} on two copies of $\ket{\psi}=c_1\ket{00}+c_3\ket{00}+c_3\ket{00}+c_4\ket{11}$ results in
\begin{equation}
    M_i\ket{\psi}^{\otimes 2}= \left(\alpha'\ket{00}+\beta'\ket{11}\right)\otimes\ket{00}\eqqcolon \ket{\psi'}\otimes\ket{00},
\end{equation}
where $\ket{\psi'}$ is an unnormalized state, $\alpha'=2a^2(c_1c_4+c_2c_3)$, and $\beta'=2b^2(c_1c_4-c_2c_3)$. The probability of success is given by the normalization constant $P'=\Tr\ket{\psi'}\bra{\psi'}$. From Theorem \ref{theo1}, we know that the optimal universal conclusive ECP applied to two copies of a renormalized $\ket{\psi'}$ produces $\phi^+$ with a success probability $P''=2|\alpha'\beta'|^2$. The full protocol applied to four copies is only successful if all steps are successful. Therefore, the final success probability is given by
\begin{align}
    P_{\psi^{\otimes 4}\to\phi^+}& = (P')^2P''\nonumber\\
        = &2^5|ab|^4\left\{\left| c_2^4 c_3^4 \right| + \left| c_1^4 c_4^4 \right| -2\mathrm{Re}\left[ c_1^2 c_4^2 (c_2^*)^2 (c_3^*)^2\right]\right\}.
\end{align}
From the domain of $f$, one can check that $2^{5}|ab|^4 \leqslant 2$, and this upper bound can be achieved when $a=b=\sqrt{2}/2$, which results in \eqref{opt4copies}.
\end{proof}

Next, we elaborate on why we focused on this concatenated structure:

First, there is a fundamental physical motivation. Because the input states have a completely unknown Schmidt basis, but the target maximally entangled state has a fixed, known basis, any universal protocol must eventually bridge this gap. Our concatenated approach explicitly isolates this requirement: the first stage transforms the unknown states into an intermediate state with a fixed, known Schmidt basis. Because we rigorously found the absolute optimal protocol for this specific basis-fixing task, our concatenation represents the optimal strategy among any protocol that enforces this intermediate step.

Second, there is a strong operational and experimental motivation. In current and near-term quantum architectures, the native entangling operations are strictly limited to two qubits. Implementing a genuine, collective 4-qubit joint LOCC measurement is experimentally challenging and require complex circuit decompositions highly susceptible to noise. Commonly used entanglement distillation protocols, such the BBPSSW~\cite{Bennett1996Jan} protocol and the DEJMPS~\cite{Deutsch1996Sep} protocol, have a similar structure to the protocol considered in this work: they operate on two copies at a time and the same operations applied repeatedly throughout the protocol. This can also be seen in several experimental realizations of entanglement distillation~\cite{reichle2006experimental,chen2017experimental,Kalb2017Jun,Hu2021Jan}. Restricting our analysis to protocols based on two-qubit operations ensures that our derived limits and optimal protocols are practically relevant and experimentally feasible.

A direct consequence of the above result is stated below.
\begin{corollary}\label{corollary_1}
If $a=b=\sqrt{2}/2$, then the map $\mathcal{E}_S$ from Theorem \ref{theo2} with Kraus operators \eqref{K_i_main} is an optimal universal conclusive ECP when acting on two copies of $\ket{\psi_S}=\alpha\ket{00}+\beta\ket{11}$. Therefore, $\mathcal{E}_S=\mathcal{E}_{\phi^+}$
\end{corollary}
\begin{proof}
The direct calculation of $\Tr \sum_iM_i \psi_S^{\otimes 2} M_i^\dagger$ with $a=b=\sqrt{2}/2$ provides the upper bound $2|\alpha\beta|^2$.
\end{proof} 

Using Corollary~\ref{corollary_1}, we can see that K\'alm\'an et al.'s protocol \cite{Kalman2025Dec} is the optimal protocol that can conclusively produce Bell pairs in this scenario, following the same scheme in Fig. \ref{fig:diagram_4_copies}. The successful branch of the map has Kraus operators $M=K_{AA'}\otimes K_{BB'}$, such that $K_{AA'}=K_{BB'}\eqqcolon K$ is given by
\begin{equation}\label{kalman_kraus}
    K=(H\otimes \ket{0}\bra{0})U_\text{CNOT}(\I\otimes \sigma_x),
\end{equation}
where $H$ is the Hadamard operator and $U_\text{CNOT}$ is the CNOT gate. The success probability of the first stage of K\'alm\'an et al.'s protocol when applied to an arbitrary pure two-qubit state \eqref{psi_c1234} is given by $P_{\text{K}}=2 \left( |c_2 c_3|^2 + |c_1 c_4|^2 \right)$ \cite{Kalman2025Dec}. The second stage of the protocol, when applied to two copies of states resulting from the first stage, yields Bell pairs with a success probability given by 
\begin{equation}
P_\text{K}'=\frac{ 
\left|(c_1 c_4)^2-(c_2 c_3)^2\right|^2 }{2\left( |c_2 c_3|^2 + |c_1 c_4|^2 \right)^2} .
\end{equation}
The entire protocol is only successful if every stage is successful. Therefore, the probability of obtaining one Bell pair from four copies of an arbitrary state \eqref{psi_c1234} is given by $\left(P_{\text{K}}\right)^2P^{\prime}_{\text{K}}$, which reproduces \eqref{opt4copies}. In fact, if we use $a=b=\sqrt{2}/2$ derived in Theorem \ref{theo3} in \eqref{K_i_main}, we reproduce \eqref{kalman_kraus}.
    
\section{Comparing universal and non-universal bounds}

When successful, a conclusive protocol always transforms the initial state $\psi$ into the target state $\phi$ using LOCCs. The maximum success probability $P_\text{Vidal}(\psi\to\phi)$ achievable by a conclusive protocol is given by Vidal's formula \cite{Vidal1999entanglement},
\begin{equation}\label{vidal}
    P_\text{Vidal}(\psi\to\phi)=\min_{l\in\{1,\ldots,n\}}\frac{E_l(\psi)}{E_l(\phi)},
\end{equation}
where $E_l(\varrho)=\sum_{i=l}^{n}\alpha_i$, for every $l\in\{1,\ldots,n\}$, are entanglement monotones written as functions of the Schmidt coefficients $\{\alpha_i\}_{i=1}^n$ of a state $\varrho$, that is,
\begin{equation}
    \ket{\varphi}=\sum_{i=1}^{n}\sqrt{\alpha_i}\ket{ii}_{AB}, \quad\alpha_i\geqslant\alpha_{i+1}\geqslant 0, \quad \sum_{i=1}^{n}\alpha_i=1.
\end{equation}
Note that Vidal's formula \eqref{vidal} provides the maximum probability $P_{\text{Vidal}}$ for converting a specific known state $\psi$ to $\phi$ but does not yield a universal protocol (one that works equally for all inputs without prior knowledge). Rather, it states that for a given pair of states, the optimal conclusive protocol achieves probability $P_{\text{Vidal}}$. In this sense, Vidal's formula does not provide the optimal \textit{universal} probability stated in Def. \ref{def:OptimalUCECP}, since it optimizes the success probability on a per-state basis.

Let us consider the input state $\ket{\psi_\lambda}=\sqrt{\lambda}\ket{00}+\sqrt{1-\lambda}\ket{11}$ such that $\lambda\in(1/2,1)$. The Schmidt coefficients $\{\alpha_i \}_{i=1}^4$ of $\psi_\lambda^{\otimes 2}$ are organized in non-increasing order as $\{\lambda^2,\lambda(1-\lambda),\lambda(1-\lambda),(1-\lambda)^2 \}$. From that, the entanglement monotones $E_l(\psi_\lambda^{\otimes 2})$ are given by
\begin{equation}\label{El_psi}
    \left\{E_l(\psi_\lambda^{\otimes 2})\right\}_{l=1}^4 = \{1,1-\lambda^2, 1-\lambda,(1-\lambda)^2 \}.
\end{equation}
Note that the target state $\phi^+$ has only two Schmidt coefficients since it is composed of only two qubits, unlike $\psi_\lambda^{\otimes 2}$. To use Vidal's formula, we can embed $\phi^+$ in a bigger space, specifically $\mc{H}_{AB}\otimes\mc{H}_{A'B'}$, by attaching it to a pure state, e.g., $\ket{00}_{A'B'}$. Thus, the embedded target state is
\begin{align}
    \ket{\phi^+}_{AB}\otimes\ket{00}_{A'B'}&=\frac{1}{\sqrt{2}}\left(\ket{0000}_{ABA'B'}+\ket{1100}_{ABA'B'}\right),\\
    &=\sum_{k=0}^3 m_k \ket{k}_{AA'}\otimes\ket{k}_{BB'},
\end{align}
where $k$ is the transformation from binary to decimal basis, $m_0=m_2=1/\sqrt{2}$, and $m_1=m_3=0$. The Schmidt coefficients $\{\beta_i \}_{i=1}^4$ of $\ket{\phi^+}\otimes\ket{00}$ are organized in non-increasing order as $\{\beta_i \}_{i=1}^4\coloneqq \left\{1/2,1/2,0,0\right\}$. Similarly, the entanglement monotones $E_l(\phi^+)$ are given by
\begin{equation}\label{El_phi+}
    \left\{E_l(\phi^+)\right\}_{l=1}^4 = \left\{1,\frac{1}{2},0,0 \right\}.
\end{equation}
Applying \eqref{El_psi} and \eqref{El_phi+} in \eqref{vidal} gives us $P_\text{Vidal}(\psi_\lambda^{\otimes 2}\to\phi^+)=\min_{i\in\{1,\ldots,4\}}\left\{1,2(1-\lambda^2),+\infty,+\infty \right\}$. Finally, the optimal probability of conclusively transforming two copies of $\psi_\lambda$ into $\phi^+$ is given by
\begin{equation}\label{Pvidal}
    P_\text{Vidal}(\psi_\lambda^{\otimes 2}\to\phi^+)=\left\{\begin{array}{ll}
        1, & \text{ if }\lambda\in(1/2,1/\sqrt{2}) \\
        2(1-\lambda^2), & \text{ if }\lambda\in[1/\sqrt{2},1) 
    \end{array} \right. .
\end{equation}

\begin{figure}[h]
    \centering
    \includegraphics[width=\linewidth]{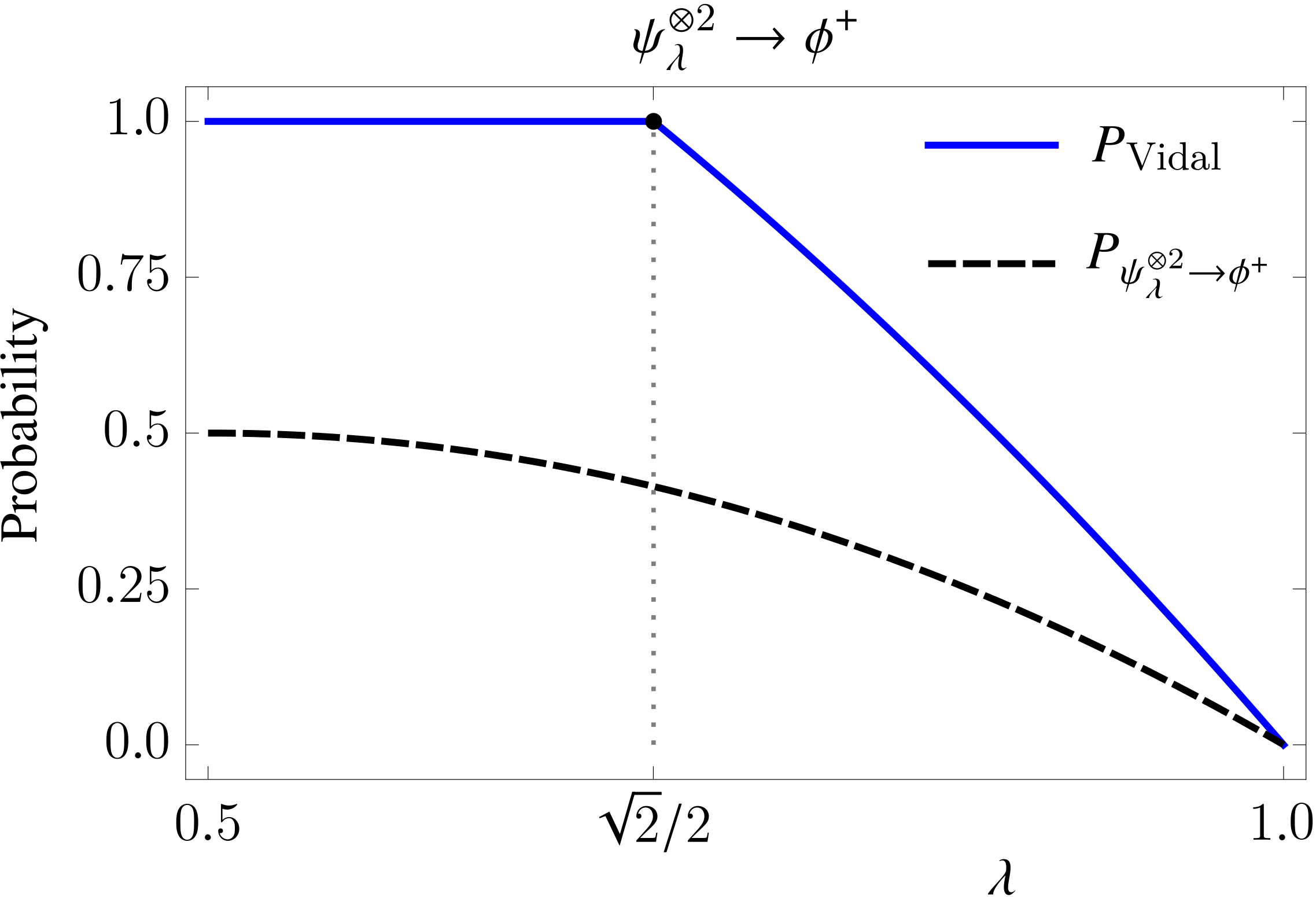}
    \caption{Probability of conclusively transforming two copies of Schmidt states $\psi_\lambda$ into a maximally entangled state $\phi^+$ as a function of $\lambda$: (blue full line) when optimized over all possible protocols given by Vidal's formula \eqref{Pvidal}; and (black dashed line) for the optimal universal protocol given by \eqref{PS2lambda}.}
    \label{fig:plotVidal}
\end{figure}

While $P_\text{Vidal}$ is optimized over all possible protocols given a specific state $\psi_\lambda$ (i.e., a specific $\lambda$), the bound from Theorem \ref{theo1} is the optimal probability that a single conclusive universal protocol can achieve for any state $\psi_\lambda$, i.e., for all $\lambda$'s. From Theorem \ref{theo1}, we have
\begin{equation}\label{PS2lambda}
    P_{\psi_\lambda^{\otimes 2}\to\phi^+}=2\lambda(1-\lambda),
\end{equation}
for $\lambda\in[1/2,1]$. In this context, if we do not know which state is being produced by the source, but we know its Schmidt basis, we can always apply the same fixed optimal universal protocol to obtain $\phi^+$ with some average probability. This average probability $\langle P_{\psi_S^{\otimes 2}\to\phi^+}\rangle$ can be obtained by integrating $P_{\psi_S^{\otimes 2}\to\phi^+}$ over $\lambda\in(1/2,1)$, following a Haar measure $f(\lambda) = 6(2\lambda - 1)^2$ (see Eq. (3.7) in Ref. \cite{Zyczkowski2001Aug}). Explicitly,
\begin{equation}
    \langle P_{\psi_\lambda^{\otimes 2}\to\phi^+}\rangle=\int_{\frac{1}{2}}^{1}P_{\psi_\lambda^{\otimes 2}\to\phi^+}f(\lambda)\text{d}\lambda=\frac{1}{5}.
\end{equation}
By applying the optimal universal protocol to an unknown source of states $\psi_\lambda$, we can obtain perfect Bell pairs with 20\% of probability. In Fig. \ref{fig:plotVidal} we present a comparison between bounds \eqref{Pvidal} and \eqref{PS2lambda}.

Now, let us assume that we do not know the Schmidt basis of the states produced by the source. If the source produces copies of an unknown pure two-qubit state, we can always apply K\'alm\'an et al.'s protocol~\cite{Kalman2025Dec} to four copies of such states to obtain one Bell pair with some probability. The expected final probability can be calculated by using moments of a Dirichlet distribution, resulting in $\langle P_{\psi^{\otimes 4}\to\phi^+}\rangle = 2/105\approx 1.9\%$ (see Appendix~\ref{appendix_C}). 

\section{Conclusion}

In this work, we have established rigorous optimality bounds for universal conclusive entanglement concentration protocols that transform multiple copies of pure two-qubit states into a maximally entangled Bell state. We focused on protocols that operate without prior knowledge of the input state's entanglement structure, a key requirement for practical applications.

Our main contributions are threefold. First, for states with a known Schmidt basis, we derived the optimal success probability ($2|\alpha\beta|^2$) for converting two copies into a Bell state using LOCC (Theorem~\ref{theo1}). Second, for arbitrary two-qubit pure states with unknown Schmidt basis, we determined the optimal probability of converting two copies into a state within a fixed Schmidt basis (Theorem~\ref{theo2}), bounded by $2(|c_1c_4| + |c_2c_3|)^2$. Third, by combining these results, we obtained the tight bound in Eq.~\eqref{opt4copies} for converting four copies of an arbitrary two-qubit pure state directly into a Bell state using concatenated protocols based on two-qubit operations (Theorem~\ref{theo3}). 

We also demonstrate that the protocol by Kalman et al.~\cite{Kalman2025Dec} saturates these bounds, proving its optimality as a universal conclusive ECP. Furthermore, we computed the expected success probability over Haar-random states, obtaining an average value of $\approx 1.9\%$, which serves as a benchmark for practical implementations. 

A critical insight from our work is the inherent trade-off between universality and efficiency: While Vidal's formula achieves higher conversion probabilities for specific input states (e.g., $P_{\text{Vidal}} = 1$ for $\lambda < 1/\sqrt{2}$), universal protocols necessarily incur a success probability reduction [e.g., $2\lambda(1-\lambda)$] due to the constraint of operating without prior state knowledge. This quantifies the fundamental cost of universality in entanglement distillation.

While our work establishes tight bounds for the direct distillation of a Bell state from four copies of an arbitrary pure state, it is important to emphasize that these fundamental limits are strictly derived within the class of concatenated LOCC protocols based on two-qubit operations. We have shown that this two-stage structure is highly motivated by both the physical necessity of fixing the unknown Schmidt basis and the practical constraints of near-term experimental hardware.

For future work, protocols that process four copies \textit{directly} without intermediate Schmidt-state conversion warrant investigation. We conjecture that such protocols cannot surpass the success probability bound in Eq.~\eqref{opt4copies}, although our two-stage approach only rigorously saturates the fundamental limits for concatenated LOCC protocols based on two-qubit operations. This implies that Schmidt-basis alignment remains an optimal strategy for universal operations. Our framework provides a foundation for extending this analysis to mixed states and larger systems, with implications for real-world quantum communication architectures.

\begin{acknowledgments}
The authors acknowledge the support from the National Science Centre Poland (Grant No. 2022/46/E/ST2/00115) and within the QuantERA II Programme (Grant No. 2021/03/Y/ST2/00178, acronym ExTRaQT, and Grant No. 2021/03/Y/ST2/00177, acronym PhoMemtor) that has received funding from the European Union's Horizon 2020 research and innovation programme under Grant Agreement No. 101017733.
   
\end{acknowledgments}

\bibliography{bibliography}

\appendix

\setcounter{theorem}{0} 

\section{}\label{appendix_A} Here, we present the detailed proof of Theorem \ref{theo1} stated in the main part of the text.

\begin{theorem}
    Following Def. \ref{def:OptimalUCECP}, the optimal universal conclusive ECP over the set of Schmidt states $\ket{\psi_S}=\alpha\ket{00}+\beta\ket{11}$ transforms two copies of such a state into $\phi^+$ with optimal probability $P_{\psi_S^{\otimes 2}\to\phi^+}\coloneqq 2|\alpha\beta|^2$.
\end{theorem}
\begin{proof}
Let $M_{i}$ be the Kraus operators of the completely positive trace non-increasing map $\mathcal{E}_1$ that is also a universal conclusive ECP. Any universal conclusive ECP $\mathcal{E}_1$ acting on pure states $\ket{\psi_S}_{AB}\otimes\ket{\psi_S}_{A'B'}$ to produce a maximally entangled state $\ket{\phi^+}$ must consist solely of Kraus operators $M_{i}$ such that 
    \begin{equation}
        M_{i}\left(\ket{\psi_S}_{AB}\otimes\ket{\psi_S}_{A'B'}\right) = \sqrt{p_{i}}\vert\phi^+\rangle_{AB}\otimes\vert\text{garb}\rangle_{A'B'},
    \end{equation}
    where $\vert\text{garb}\rangle_{A'B'}$ is some product garbage state and $p_{i}\in[0,1]$.

    Since we are considering only LOCC protocols, product states can not produce any entanglement, regardless of the universal conclusive ECPs. Therefore, the Kraus operator $K_{i}$ must be such that:
    \begin{subequations}
    \begin{align}
        M_{i}\vert00\rangle_{AB}\vert00\rangle_{A'B'} &= 0\\
        M_{i}\vert11\rangle_{AB}\vert11\rangle_{A'B'} &= 0.
    \end{align}
    \end{subequations}
    From the constraints above, the action of $\mathcal{E}_{1}$ on $\ket{\psi_S}^{\otimes 2}$ is:
    \begin{equation}
        M_{i}\ket{\psi_S}^{\otimes 2}=\sqrt{2}\alpha\beta\,\, M_{i}\left(\frac{\ket{0011}+\ket{1100}}{\sqrt{2}}\right).
    \end{equation}
    Hence, the probability of success is given by
    \begin{equation}
        \operatorname{Tr}\left[\mathcal{E}_1\left(\psi_{S}^{\otimes 2}\right)\right] = \operatorname{Tr}\left[\sum_{i}p_{i}\left(2|\alpha\beta|^2\right) \phi^+\otimes \text{garb}\right]\leqslant 2|\alpha\beta|^2,
    \end{equation}
    where $p_i$ is the probability of converting $(\ket{0011}+\ket{1100})/\sqrt{2}$ into $\phi^+$. The inequality above implies that the maximal probability of obtaining $\phi^+$ from $\vert\psi_{S}\rangle^{\otimes 2}$ using a universal conclusive ECP is bounded by $2|\alpha\beta|^2$. Therefore, following Def. \ref{def:OptimalUCECP}, to be considered an optimal universal conclusive ECP from two copies of Schmidt states, a protocol must produce $\phi^+$ with optimal universal probability $P_{\psi_S^{\otimes 2}\to\phi^+}\coloneqq 2|\alpha\beta|^2$.
\end{proof}

\section{}\label{appendix_B} Below, we present the detailed proof of Theorem \ref{theo2} stated in the main part of the text.

\begin{theorem}
    To produce arbitrary states in a known Schmidt basis from two copies of arbitrary states $\ket{\psi}_{AB}$ given by \eqref{psi_c1234}, a map $\mathcal{E}_S(\cdot)=\sum_i M_i(\cdot) M_i^\dagger$ must be composed of Kraus operators $M_i=K_i\otimes K_i$ acting on $\mathcal{H}_{AA'}\otimes\mathcal{H}_{BB'}$ such that,
    \begin{equation}
        K_i=a(\ket{00}\bra{01} + \ket{00}\bra{10})+b(\ket{10}\bra{01} - \ket{10}\bra{10})
    \end{equation}
    for every $i$, where $a$ and $b$ are any non-zero complex numbers satisfying $2(|a|^4+|b|^4)\leqslant 1$. The success probability $P_{\psi^{\otimes 2}\to\psi_S}$ of such a protocol is bounded by $\overline{P}_{\psi^{\otimes 2}\to\psi_S}=2(|c_1c_4|+|c_2c_3|)^2$, i.e., $P_{\psi^{\otimes 2}\to\psi_S}<\overline{P}_{\psi^{\otimes 2}\to\psi_S}$.
\end{theorem}
\begin{proof}
     From \eqref{systemM2}, \eqref{M_ipsi1} reduces to
    \begin{equation}\label{Mipsipsi=Miphi1phi2}
    M_i \left(\ket{\psi}^{\otimes 2} \right)= M_i [\sqrt{2}c_1c_4\ket{\phi_1}+\sqrt{2}c_2c_3\ket{\phi_2}].
    \end{equation}
    where we introduced the notation $\ket{\phi_1}\coloneqq(\ket{0011}+\ket{1100})/\sqrt{2}$ and $\ket{\phi_2}\coloneqq(\ket{0110}+\ket{1001})/\sqrt{2}$. Now, let us consider the situation when $c_2=c_3=0$, which results in $M_i(\ket{\psi}^{\otimes 2})=\sqrt{2}c_1c_4M_i\ket{\phi_1}$. Therefore, $\sum_i M_i \psi^{\otimes 2} M_i^\dagger=2|c_1c_4|^2 \sum_iM_i\phi_1 M_i^\dagger$, which, by definition \eqref{Lambda_S}, means that
    \begin{equation}\label{MphiM=Ppsi}
        2|c_1c_4|^2 \sum_iM_i\phi_1 M_i^\dagger = N_1\psi_S^{(1)},
    \end{equation}
    where $\ket{\psi_S^{(1)}}=\alpha_1\ket{00}+\beta_1\ket{11}\in\mathcal{H}_{AB}$ for some pair of coefficients $\alpha_1,\beta_1$ and $N_1$ is a normalization constant. Since both $\phi_1$ and $\psi_S^{(1)}$ are pure states, we have that \eqref{MphiM=Ppsi} implies that 
    \begin{equation}\label{Miphi1}
        M_i\ket{\phi_1}=\sqrt{p_{i,1}}\ket{\psi_S^{(1)}},
    \end{equation}
    for every $i$, where $p_{i,1}$ is a probability distribution satisfying $\sum_i p_{i,1}=N_1/2|c_1c_4|^2$. Note that this last equation implies that $N_1\leqslant 2|c_1c_4|^2$. A similar argument provides
    \begin{equation}\label{Miphi2}
        M_i\ket{\phi_2}=\sqrt{p_{i,2}}\ket{\psi_S^{(2)}}
    \end{equation}
    for every $i$, where $p_{i,2}$ is a probability distribution satisfying $\sum_i p_{i,2}=N_2/2|c_2c_3|^2$. As before, we also have that $N_2\leqslant 2|c_2c_3|^2$. Now, using \eqref{Miphi1} and \eqref{Miphi2} in \eqref{Mipsipsi=Miphi1phi2} gives us
    \begin{align}
        M_i \ket{\psi}^{\otimes 2} =&\, c_1c_4\sqrt{2p_{i,1}}\ket{\psi_S^{(1)}}+c_2c_3\sqrt{2p_{i,2}}\ket{\psi_S^{(2)}}\nonumber\\
        =& \left(c_1c_4\sqrt{2p_{i,1}}\alpha_1+c_2c_3\sqrt{2p_{i,2}}\alpha_2\right) \ket{00}\nonumber\\
        &+\left(c_1c_4\sqrt{2p_{i,1}}\beta_1+c_2c_3\sqrt{2p_{i,2}}\beta_2\right) \ket{11}.
    \end{align}
    Finally, from \eqref{Lambda_S}, the probability $P_{\psi^{\otimes 2}\to\psi_S}=\Tr \sum_i M_i\psi^{\otimes 2} M_i^\dagger$ is expressed as
    \begin{align}
        P_{\psi^{\otimes 2}\to\psi_S}=&\,2\sum_i\left(\left|c_1c_4\sqrt{p_{i,1}}\alpha_1+c_2c_3\sqrt{p_{i,2}}\alpha_2\right|^2\right.  \nonumber\\ &+ \left. \left|c_1c_4\sqrt{p_{i,1}}\beta_1+c_2c_3\sqrt{p_{i,2}}\beta_2\right|^2\right).\label{PSpsi}
    \end{align}
    Now, let us write each Kraus operator $K_i$ in \eqref{Mi=KiKi} using a Pauli basis as follows
    \begin{equation}\label{Kpauli}
        K_i=\sum_{k,l=1}^{4} r_{kl} \sigma_k\otimes\sigma_l
    \end{equation}
    where $(\sigma_1,\sigma_2,\sigma_3,\sigma_4)=(\sigma_x,\sigma_y,\sigma_z,\I)$ and $r_{kl}^i\in\mathbb{R}_+$ for every $k,l,i$. In principle, $r_{kl}$ could be distinct for every $i$, but we are going to omit this extra index unless necessary. By using \eqref{Kpauli} in system \eqref{systemM}, we obtain, after some algebra,
    \begin{equation}
    \begin{array}{llll}\label{conditions}
        r_{21} = -r_{12},& r_{22} = r_{11},&r_{23} = \im r_{14},&r_{24} = \im r_{13},\\
        r_{41} = -\im r_{32},& r_{42} = \im r_{31}, & r_{43} = -r_{34}, & r_{44} = -r_{33},
    \end{array}
    \end{equation}
    for every $i$. We can insert Eqs. \eqref{conditions} into \eqref{Kpauli} and solve \eqref{Miphi1} and \eqref{Miphi2}. Note that the auxiliary state on the right-hand side of \eqref{Miphi1} and \eqref{Miphi2} was omitted. Without loss of generality and to make the equations more tractable, we can suppose that $\ket{\text{aux}}=\ket{00}_{A'B'}$. Also, both $\ket{\psi_S^{(1,2)}}$ are vectors such that 14 out of 16 entries are null. The resolution of these 14 equations yields 8 equivalent solutions, up to a complex phase that can be accommodated by the remaining free coefficients. One of those solutions is provided below:
    \begin{equation}
        \begin{array}{lll}
        r_{11} = r_{34}, & r_{12} = \im r_{34}, & r_{13} = r_{14},\\
        r_{31} = r_{14}, & r_{32} = \im r_{14}, & r_{33} = r_{34}.
        \end{array}
    \end{equation}
    From the relations above and \eqref{conditions}, the Kraus operators $K_i$ can be expressed in terms of only two complex parameters, namely, $r_{14}\coloneqq a/4$ and $r_{34}\coloneqq b/4$, as
    \begin{equation}\label{KiR1R2}
        K_i=a(\ket{00}\bra{01} + \ket{00}\bra{10})+b(\ket{10}\bra{01} - \ket{10}\bra{10}).
    \end{equation}
    Remember here that $M_i$ acts on $\mathcal{H}_{AA'}\otimes\mathcal{H}_{BB'}$ such that $M_i=K_i\otimes K_i$, where each $K_i$ acts locally and identically on each laboratory, i.e., $AA'$ and $BB'$.
    Finally, by using \eqref{KiR1R2} in \eqref{Miphi1} and \eqref{Miphi2}, we obtain
    $\sqrt{p_{i,j}}\alpha_{j}=\sqrt{2}a^2$ and $\sqrt{p_{i,j}}\beta_{j}=(-1)^{j-1}\sqrt{2}b^2$ for $j=1,2$. By definition, $|\alpha_j|^2+|\beta_j|^2=1$ and $0\leqslant p_{i,j}\leqslant 1$ for $j=1,2$. Therefore, coefficients $a$ and $b$ must satisfy 
    \begin{equation}\label{constrain}
       2(|a|^4+|b|^4)\leqslant 1. 
    \end{equation}
    In addition, states $\ket{\psi_S^{(1,2)}}$ must be entangled, which requires $\alpha_{1},\beta_1$ or $\alpha_2,\beta_2$ to be nonzero. That implies that both $a$ and $b$ must be nonzero.
    
    Now, let us obtain the upper bound for $P_{\psi^{\otimes 2}\to\psi_S}$. In \eqref{PSpsi}, define $A_i=c_1c_4\sqrt{p_{i,1}}\alpha_1+c_2c_3\sqrt{p_{i,2}}\alpha_2$ and $B_i=c_1c_4\sqrt{p_{i,1}}\beta_1+c_2c_3\sqrt{p_{i,2}}\beta_2$. Therefore
    \begin{align}\label{A+B}
        |A_i|^2+|B_i|^2 =&\, |c_1c_4|^2p_{i,1}+|c_2c_3|^2p_{i,2}\nonumber\\
        &+2\text{Re}\left(c_1c_4c_2^*c_3^*g \right)\sqrt{p_{i,1}p_{i,2}},
    \end{align}
    where $g\coloneqq \alpha_1\alpha_2^*+\beta_1\beta_2^*$. Since for any complex number $z$, it is always true that $\text{Re}(z)\leqslant|z|$, we have
    \begin{equation}
        \text{Re}\left(c_1c_4c_2^*c_3^* g \right)\leqslant|c_1c_4||c_2c_3| |g|.
    \end{equation}
    Observe also that $|g|=f/\sqrt{p_{i,1}p_{i,2}}$, where $f={2}\left| |a|^4-|b|^4\right|$.  From \eqref{constrain}, we have that $f\in[0,1]$. In particular, $f=1$ iff $|a|=\frac{1}{\sqrt[4]{2}}$ and $b=0$ or the opposite, which results in unentangled Schmidt states. Therefore, we must impose $f<1$, which makes \eqref{A+B} become
    \begin{align}
        \sum_i|A_i|^2+|B_i|^2 <&\, |c_1c_4|^2\left(\sum_ip_{i,1}\right)+|c_2c_3|^2\left(\sum_ip_{i,2}\right)\nonumber\\
        &+2|c_1c_4||c_2c_3|.
    \end{align}
    Since $\sum_i p_{i,j}\leqslant 1$ for $j=1,2$, we obtain
    \begin{equation}\label{bound_schmidt}
        P_{\psi^{\otimes 2}\to\psi_S} < 2\left(|c_1c_4|+|c_2c_3| \right)^2\eqqcolon \overline{P}_{\psi^{\otimes 2}\to\psi_S}. 
    \end{equation}
    Consequently, the maximum value of $P_{\psi^{\otimes 2}\to\psi_S}$ is $1/2$.
    \end{proof}

\section{}\label{appendix_C}
    Here we calculate the expected final probability $\langle P_{\psi^{\otimes 4}\to\phi^+}\rangle$. When the source produces copies of an unknown pure two-qubit state, we can always apply K\'alm\'an et al.'s protocol to a set of four states to obtain one Bell pair with some probability. Like before, this average probability can be obtained by integrating \eqref{opt4copies} over $\{c_i\}_{i=1}^4$ using a Haar measure. Under Haar measure, the squared amplitudes $x_i=|c_i|^2$ follow a Dirichlet distribution $\textrm{Dir}(1,1,1,1)$ \cite{bengtsson2017geometry}, while the phases $\theta_i=\arg(c_i)$ are independent and uniformly distributed in the interval $[0,2\pi]$. First, let us show that the average of the last term in \eqref{opt4copies} is zero. Since $c_1^2 c_4^2 (c_2^*)^2 (c_3^*)^2=x_1x_2x_3x_4e^{\im 2(\theta_1+\theta_4-\theta_2-\theta_3)}$, then
\begin{equation}
    \mathrm{Re}\left[ c_1^2 c_4^2 (c_2^*)^2 (c_3^*)^2\right]=x_1x_2x_3x_4\cos\eta.
\end{equation}
where we define $\eta=2(\theta_1+\theta_4-\theta_2-\theta_3)\mod 2\pi$, which is uniformly distributed in $[0,2\pi)$. Let us fix $\mathbf{x}=(x_1,x_2,x_3,x_4)$. A direct integration provides the conditional expectation $\mathbb{E}[\cos\eta|\mathbf{x}]=0$. Because $|c_i^4|=|c_i|^4=x_i^2$, the final expectation value is given by 
\begin{equation}
    \mathbb{E}\left[P_\mathrm{K}^\text{final}\right]=\mathbb{E}\left[\mathbb{E}[P_\mathrm{K}^\text{final}|\mathbf{x}]\right]=2\mathbb{E}[x_1^2x_4^2]+2\mathbb{E}[x_2^2x_3^2].
\end{equation}
By symmetry of the Dirichlet distribution $\text{Dir}(1,1,1,1)$, we have that $\mathbb{E}[x_1^2x_4^2]=\mathbb{E}[x_2^2x_3^2]$. The general formula for the moments of the distribution $\text{Dir}(\boldsymbol{\alpha})$ is given by \cite{Schiavo2019Jan}
\begin{equation}\label{moments}
    \mathbb{E}\left[\prod_{i} x_i^{\beta_i} \right]=\frac{\Gamma\left( \sum_i \alpha_i\right)}{\Gamma\left( \sum_i \alpha_i+\sum_i \beta_i\right)}\prod_i\frac{\Gamma(\alpha_i+\beta_i)}{\Gamma(\alpha_i)},
\end{equation}
where the Gamma function for integers is simply $\Gamma(n)=(n-1)!$. To calculate $\mathbb{E}(x_1^2x_4^2)$ we must use $(\alpha_1,\alpha_2,\alpha_3,\alpha_4)=(1,1,1,1)$ and $(\beta_1,\beta_2,\beta_3,\beta_4)=(2,0,0,2)$. Having said that, from \eqref{moments}, we obtain $\mathbb{E}(x_1^2x_4^2)=1/210$. Identical arguments give us $\mathbb{E}(x_2^2x_3^2)=1/210$, which results in the expected final probability $\langle P_{\psi^{\otimes 4}\to\phi^+}\rangle = 2/105$. Numerically, we have obtained $\langle P_{\psi^{\otimes 4}\to\phi^+}\rangle \sim 0.0190$ for $10^4$ initial states \eqref{psi_c1234} where the real and imaginary parts of each $c_i$ follow a Gaussian probability distribution with mean zero and standard deviation $1/\sqrt{2}$.

\end{document}